\documentclass[11pt]{article}

\usepackage[margin=1in]{geometry}

\usepackage{amsfonts,amsmath,amsthm,amssymb}
\usepackage{graphicx}
\usepackage{enumitem}
\usepackage{hyperref}

\newcommand{\R}{\mathbb R}

\DeclareMathOperator{\interior}{int}

\newtheorem{theorem}{Theorem}
\newtheorem{lemma}[theorem]{Lemma}
\newtheorem{corollary}[theorem]{Corollary}

\theoremstyle{remark}
\newtheorem*{remark}{Remark}

\title{On the complexity of the free space of a translating box in $\R^3$}
\author{Gabriel Nivasch\thanks{Department of Computer Science, School of Computer Science, Ariel University, Ariel, Israel. \texttt{gabrieln@ariel.ac.il}}}
\date{}

\begin{document}
	\maketitle
	
\begin{abstract}		
	Consider a convex polyhedral robot $B$ that can translate (without rotating) amidst a finite set of non-moving polyhedral obstacles in $\R^3$. The \emph{free space} $\mathcal F$ of $B$ is the set of all positions in which $B$ is disjoint from the interior of every obstacle.
		
	Aronov and Sharir (1997) derived an upper bound of $O(n^2\log n)$ for the combinatorial complexity (number of vertices, edges, and faces) of $\mathcal F$, where $n$ is the total complexity of the obstacles, and the complexity of $B$ is assumed constant.
		
	Halperin and Yap (1993) showed that, if $B$ is either a box or a ``flat'' convex polygon, then a tighter bound of $O(n^2\alpha(n))$ holds.  Here $\alpha(n)$ is the inverse Ackermann function.
		
	In this paper we prove that if $B$ is a box, then the complexity of $\mathcal F$ is $O(n^2)$. Furthermore, if $B$ is a convex polygon whose edges come in parallel pairs, then the complexity of $\mathcal F$ is $O(n^2)$ as well. These results settle the question of the asymptotic worst-case complexity of $\mathcal F$ for a box, as well as for all convex polygons.
		
	Keywords: motion planning, computational geometry, Minkowski sum, lower envelope
\end{abstract}
	
\section{Introduction}\label{sec_intro}
	
One of the most basic problems studied in algorithmic motion planning is that of moving a rigid ``robot'' $B$ among fixed obstacles. There are many possible instances of the problem, depending on the dimension of the ambient space ($\R^2$, $\R^3$, etc.), the shape of the robot, and the type of motion allowed (e.g., whether the robot is allowed to rotate or just to translate).
	
Usually we are given an initial placement of the robot, as well as a desired final placement, and the objective is to find a continuous motion that takes the robot from the initial placement to the final one, while avoiding the obstacles at all times (or determine that no such motion exists).
	
Each placement of the robot can be parametrized by a $k$-tuple of real numbers, where $k$ is the number of degrees of freedom of the robot. The \emph{configuration space} is thus a $k$-dimensional space, each point of which corresponds to a placement of the robot. The \emph{free space} of the robot is the subset of the configuration space that corresponds to all placements that are \emph{free} (or \emph{legal}), in the sense that the robot does not intersect any obstacle. In the case where both the robot and the obstacles are polyhedra in $\R^3$ and the robot can only translate, the configuration space is also three-dimensional, and the free space has polyhedral boundaries, consisting of vertices, edges, and faces.
	
A basic parameter in the analysis of many motion planning algorithms is the \emph{combinatorial complexity} of the free space, meaning the total number of vertices, edges and faces in the boundary of the free space. Hence, a natural question is to determine the worst-case complexity of the free space in different scenarios. See the survey chapter by Halperin, Salzman, and Sharir~\cite{handbookChapter} for more background on algorithmic motion planning.
	
In this paper we consider the case in which the robot is either a ``flat'' convex polygon or a convex polyhedron that can only translate, and the obstacles are polyhedra in $\R^3$. We denote the total complexity (number of vertices, edges, and faces) of the obstacles by $n$. The complexity of the robot is assumed constant.
	
As will be explained below in more detail, the complexity of $\mathcal F$ is asymptotically determined by the number of triple contacts between the robot and obstacles, i.e., the number of free placements in which three distinct robot features (vertices, edges, or faces) simultaneously make contact with three distinct obstacle features.
	
Hence, a trivial upper bound for the complexity of the free space is $O(n^3)$. Aronov and Sharir \cite{AS} improved the upper bound to $O(n^2\log n)$. Halperin and Yap \cite{HY} studied the case where the robot is a ``flat'' convex polygon, as well as the case where the robot is a three-dimensional axis-parallel box (or simply, a ``box''). For both these cases they derived an upper bound of $O(n^2\alpha(n))$, where $\alpha(n)$ is the very slow-growing inverse Ackermann function. The inverse-Ackermann factor arises from consideration of upper and lower envelopes of segments.
	
Halperin and Yap \cite{HY} also proved that if the robot is a rectangle and the obstacles are $n$ \emph{lines}, then the complexity of $\mathcal F$ is $O(n^2)$.

\begin{figure}
	\centering\includegraphics[width=14cm]{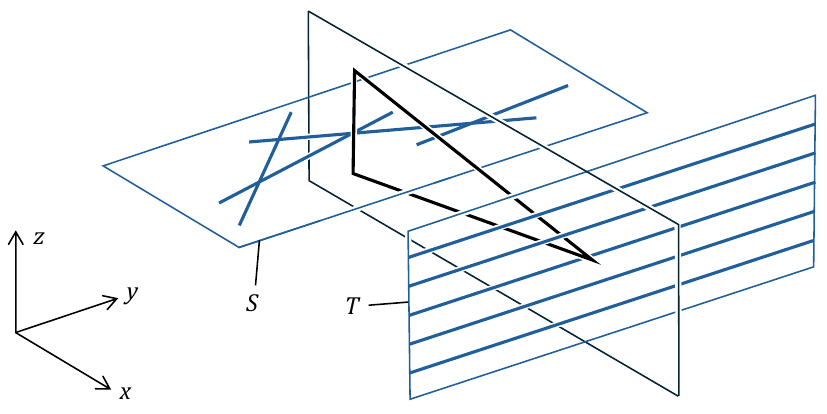}
	\caption{\label{fig_triangle_lower_bound}Configuration that achieves $\Omega(n^2\alpha(n))$ triple contacts with a triangular robot.}
\end{figure}
	
Regarding lower bounds for the complexity of $\mathcal F$, a lower bound of $\Omega(n^2)$ is easy to obtain for any robot. For some types of robots (e.g., a triangle) one can obtain a lower bound of $\Omega(n^2\alpha(n))$, as follows (slightly adapted from \cite{AS}):

Let the triangular robot $B$ be parallel to the $xz$-plane, with one vertical edge (parallel to the $z$-axis). Place a configuration $S$ of $n/2$ obstacle edges parallel to the $xy$-plane that form $\Omega(n\alpha(n))$ lower-envelope intersections. Compress $S$ in the $x$-direction, so that the $x$-coordinates of the intersections are not too far from each other. In addition, give the edges of $S$ slightly different $z$-coordinates, in order to make them pairwise disjoint. Place another set $T$ of $n/2$ obstacle edges parallel to the $y$-axis with different $z$-coordinates, close to the third vertex of $B$. There are $\Omega(n \alpha(n))$ ways the vertical edge of $B$ can make a double contact with two edges of $S$. For each such possibility, the $z$-coordinate of $B$ can be adjusted so as to make $\Omega(n)$ different third contacts with the edges of $T$.
See Figure \ref{fig_triangle_lower_bound}.

For the two-dimensional case, in which $B$ is a convex polygon free to translate in $\R^2$ and the obstacles are polygons with a total of $n$ vertices, the complexity of $\mathcal F$ is $O(n)$ (Kedem, Livne, Pach, and Sharir~\cite{klps86}).

Note that, by affine transformations, the cases where $B$ is a cube or a box or a parallelepiped are all equivalent.
	
\subsection{Our results}
	
In this paper we prove the following:

\begin{theorem}\label{thm_box}
	Let $B$ be a box-shaped robot that is free to translate in $\R^3$ amidst polyhedral obstacles of total complexity $n$. Then the complexity of the free space of $B$ is $O(n^2)$.
\end{theorem}

\begin{figure}
\centering\includegraphics[width=6cm]{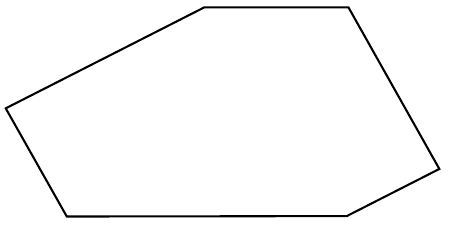}
\caption{\label{fig_fully_parallel} A fully-parallel polygon.}
\end{figure}
	
We also study convex polygonal robots in more depth. As mentioned before, there is an upper bound of $O(n^2\alpha(n))$ for any convex polygonal robot, as well as a matching lower bound of $\Omega(n^2\alpha(n))$ for \emph{some} such robots (e.g., a triangle).\footnote{This is how the claimed bound of $\Theta(n^2\alpha(n))$ for convex polygonal robots in \cite{HY} is to be understood.} In this paper we provide a complete characterization of the asymptotic worst-case complexity of $\mathcal F$ for every specific convex polygon.

We call a convex polygon \emph{fully parallel} if its edges come in parallel pairs (not necessarily of the same length). See Figure \ref{fig_fully_parallel}.

It is easy to see that the lower bound of $\Omega(n^2\alpha(n))$ shown in Figure \ref{fig_triangle_lower_bound} can be achieved with any polygonal robot that is not fully parallel: Any edge of $B$ that does not have a parallel edge on the other side can serve as the vertical robot edge of Figure \ref{fig_triangle_lower_bound}. The lower bound of $\Omega(n^2\alpha(n))$ can also be achieved with any polyhedral robot that can be projected to form a non-fully-parallel polygon: Since our problem is invariant under affine transformations, we can first transform $B$ into an almost-flat non-fully-parallel polygon, so that its thickness is negligible, and then add the obstacles as explained before.

We prove the following:

\begin{theorem}\label{thm_fully-parallel}
	Let $B$ be a fully-parallel convex polygon that is free to translate in $\R^3$ amidst polyhedral obstacles of total complexity $n$. Then the complexity of the free space of $B$ is $O(n^2)$ (where the hidden constant depends on $B$).
\end{theorem}

The implicit constant in Theorem \ref{thm_fully-parallel} is proportional to the cube of the number of vertices of $B$, and it also depends on the maximum ratio between lengths of parallel edges of $B$.

\section{Preliminaries}\label{sec_preliminaries}

Let $B\subset \R^3$ be a polyhedral robot that can translate among a set $\mathcal C = \{C_1,\ldots, C_k\}$ of pairwise-disjoint polyhedral obstacles in $\R^3$. A placement of $B$ in space can be specified by a point $v$ in \emph{configuration space} $\R^3$. At that placement, the robot occupies the points $B+v=\{b+v:b\in B\}$. Such a placement is \emph{free} (or \emph{legal}) if $B+v$ is disjoint from the interior of every obstacle in $\mathcal C$. If at that placement the robot makes contact with one or more obstacles, i.e., it intersects their boundary but not their interior, then the placement is called \emph{semi-free}. The set of all points in configuration space that correspond to free placements of the robot is called the \emph{free space} of the robot, which we denote by $\mathcal F$. The set of semi-free placements constitute the boundary $\partial \mathcal F$. This boundary can be partitioned into vertices, relatively open edges, and relatively open faces, and we wish to bound the total number of these features.

An observation dating back to Lozano-P\'erez and Wesley \cite{lozano1979} is that the robot can be replaced by a point if the obstacles are appropriately expanded. For every obstacle $C_i\in\mathcal C$, define the expanded obstacle $D_i=C_i\oplus (-B)$, where $X\oplus Y = \{x+y:x\in X, y\in Y\}$ denotes the \emph{Minkowski sum} of two sets of points. Then the free space is given by
\begin{equation}\label{eq_expanded}
	\mathcal F = \R^3\setminus \bigcup_{i=1}^k \interior(D_i).
\end{equation}
Thus, $\mathcal F$ is the complement of the union of a set of possibly overlapping polyhedra. Aronov, Sharir, and Tagansky \cite{AST} study the complexity of the union of possibly overlapping convex polyhedra, although here we are interested in a particular case in which the polyhedra are of a specific form (and are not necessarily convex).

The expanded obstacles $D_i$ induce a partition of $\R^3$ into relatively open cells of dimensions $0,\ldots,3$, as follows. For each expanded obstacle $D_i$, let $\mathcal D_i$ be the partition of $D_i$ into its relative interior, its relatively open faces, its relatively open edges, and its vertices. Let $\mathcal D = \bigcup \mathcal D_i$. Then each cell is a connected component of a set of the form $(\bigcap \mathcal L) \setminus \bigcup (\mathcal D\setminus \mathcal L)$ for some $\mathcal L\subseteq \mathcal D$, meaning a maximal connected set of points that belong to some fixed elements of $\mathcal D$, and only to those elements. The interior of $\mathcal F$ is obtained by taking $\mathcal L=\emptyset$ (in which case $\bigcap\mathcal L = \R^3$). The cells of $\partial\mathcal F$ are those that belong to the boundary of at least $D_i$ but do not belong to the interior of any $D_i$, and they have dimensions $0,1,2$.

\paragraph{General position assumption.} It can be shown \cite{AS,AST} that for the problem of asymptotically bounding the maximum complexity of $\mathcal F$, we can assume that the robot is in general position with respect to the obstacles. The proof of this fact is a bit delicate, and we present it in Section \ref{sec_genpos} for completeness.

General position entails, among other things, that no obstacle face is parallel to the line joining any two robot vertices, and that no robot face is parallel to the line joining any two obstacle vertices. In addition: no four faces of distinct expanded obstacles intersect at a common point; no vertex of one expanded obstacle lies in the boundary of another one; no two of edges of distinct expanded obstacles intersect; and no edge of an expanded obstacle meets the intersection of the boundaries of two other expanded obstacles. We will make some additional general-position assumptions during the course of the arguments in this paper.

\paragraph{Types of contacts.} A \emph{contact specification} is a pair $O=(f, g)$ where $f$ is a feature (vertex, edge, or face) of the robot, and $g$ is a feature of an obstacle, such that the robot can be translated so as to make any given point of $f$ coincide with any given point of $g$ without $g$ penetrating the robot. Due to the general position assumption, we must have $\dim f + \dim g\le 2$. A (not necessarily free) placement of the robot is said to \emph{realize} the contact specification $O$ if at that placement $f$ makes contact with $g$.

We call the contact specification $O=(f,g)$ \emph{generic}, \emph{singly degenerate}, or \emph{doubly degenerate}, according to whether $\dim f + \dim g$ equals $2$, $1$, or $0$, respectively. If $O$ is generic, then we call it a \emph{vertex contact}, an \emph{edge contact}, or a \emph{face contact}, according to whether the robot feature $f$ is a vertex, an edge, or a face, respectively. See Figure \ref{fig_triple_contact}.

Due to the general position assumption, each face of $ \mathcal F$ arises from a connected set of robot placements that realize a fixed generic contact specification. (The same contact specification can give rise to several disconnected, coplanar faces of $ \mathcal F$.)

Some edges of $ \mathcal F$ arise from the intersection of two faces of expanded obstacles; they correspond to connected sets of robot placements that simultaneously realize two generic contact specifications. Other edges of $\mathcal F$ are part of an edge of an extended obstacle; they realize a singly-degenerate contact specification. (The same set of contact specifications can give rise to several disconnected, collinear edges of $ \mathcal F$.)

Some vertices of $ \mathcal F$ arise from the intersection of three faces of extended obstacles; each such vertex corresponds to the robot simultaneously realizing three generic contact specifications (see Figure \ref{fig_triple_contact} again); we call it a \emph{triple-contact vertex}. Other vertices of $\mathcal F$ arise from the intersection between an expanded-obstacle edge and an expanded-obstacle face; each such vertex corresponds to the robot simultaneously realizing a singly-degenerate contact specification and a generic contact specification; we call it a \emph{double-contact vertex}. Other vertices of $\mathcal F$ are vertices of expanded obstacles; each such vertex corresponds to the robot realizing a doubly-degenerate contact specification; we call it a \emph{single-contact vertex}.

\begin{figure}
	\centering\includegraphics[width=9cm]{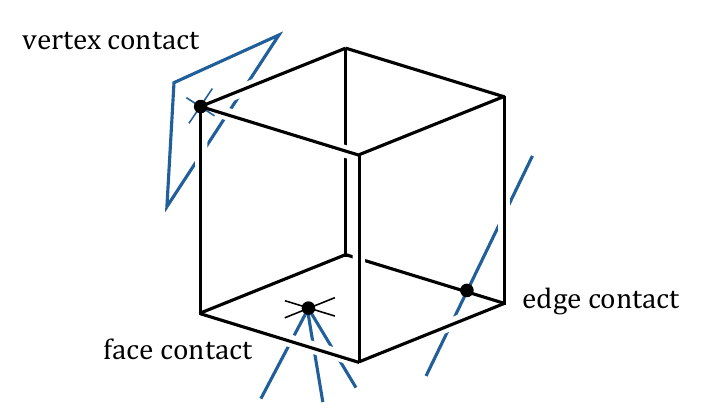}
	\caption{\label{fig_triple_contact} A placement of $B$ (here a cube) making three contacts with obstacles, in this case a vertex contact, an edge contact, and a face contact. This placement of $B$ corresponds to a vertex of the free space.}
\end{figure}

\paragraph{\boldmath Triple-contact vertices bound the complexity of $\mathcal F$.}
Following \cite{AS,HY}, we argue that in order to asymptotically bound the complexity of $\mathcal F$, it is enough to bound its number of triple-contact vertices.

The number of single-contact vertices is $O(n)$, and the number of double-contact vertices is $O(n^2)$. Hence, the only vertices left to bound are the triple-contact vertices.

Let us now bound the number of edges and faces using charging arguments. Let $V$ be the number of triple-contact vertices. By the general position assumption, triple-contact and double-contact vertices have degree $3$. Consider a single-contact vertex $v\in\partial\mathcal F$ specified by $O=(f,g)$, where $f$ is a robot vertex and $g$ is an obstacle vertex. The degree of $v$ is at most the sum of the degrees of $f$ and $g$. Give each edge of $\mathcal F$ two tokens, and let each edge pass one token to each of its endpoint vertices. Then each vertex receives a number of tokens equal to its degree, so the total number of tokens is at most $3V + O(n^2)$. Hence, the number $E$ of edges of $\mathcal F$ satisfies $E\le O(V+n^2)$.

Finally, give each face of $\mathcal F$ a token, and let each face pass its token to one of its boundary edges. Then each edge receives at most two tokens, so the number $F$ of faces of $\mathcal F$ satisfies $F\le 2E$.

\section{Upper and lower envelopes}\label{sec_lower_env}

In this section we present some results, some of them new, on upper and lower envelopes, which we will need in subsequent sections.

Let $S$ be a finite collection of (possibly intersecting) graphs of total or partial functions $\R\to\R$. The \emph{lower envelope} of $S$ is the graph of the pointwise minimum of these functions; it consists of the parts of the function graphs that are visible from below at $y=-\infty$. Similarly, the \emph{upper envelope} of $S$ consists of the parts of the function graphs that are visible from above at $y=+\infty$.

\begin{figure}
	\centering\includegraphics[width=8cm]{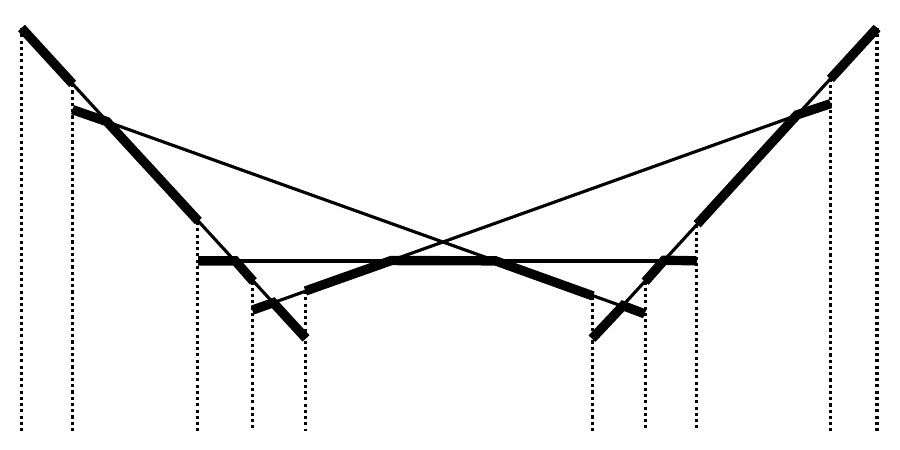}
	\caption{\label{fig_lower_env}The lower envelope of a set of segments, decomposed into concave chains.}
\end{figure}

The case that concerns us in this paper is the one in which $S$ is a collection of $n$ nonvertical line segments. In this case, the lower envelope of $S$ consists of at most $2n-1$ concave chains, where the endpoints of the concave chains are the segment endpoints that are visible from below (see Figure \ref{fig_lower_env}). Similarly, the upper envelope of $S$ consists of at most $2n-1$ convex chains.

Hart and Sharir \cite{HS} proved that, if $S$ is a collection of nonvertical line segments, then the combinatorial complexity of each envelope of $S$ is at most $O(n\alpha(n))$, where $\alpha(n)$ is the inverse-Ackermann function. Furthermore, this bound is worst-case tight (Wiernik and Sharir \cite{WS}).

Halperin and Yap (Lemma 2.3 in \cite{HY}) showed that if $S_1$ and $S_2$ are two collections of nonvertical line segments in the plane, with $|S_1\cup S_2|=n$ and no two segments of $S_1\cup S_2$ overlapping at an interval, then each envelope (upper or lower) of $S_1$ intersects each envelope (upper or lower) of $S_2$ in at most $O(n\alpha(n))$ points. We refine this result:

\begin{lemma}\label{lem_env}
	Let $S$ be a collection of $n$ nonvertical line segments in the plane with no two segments overlapping at an interval, and let $S_1$, $S_2$ be a partition of $S$ into two parts. Then:
	\begin{enumerate}
		\item The lower envelopes of $S_1$ and $S_2$ (and, similarly, the upper envelopes of $S_1$ and $S_2$) intersect one another in at most $O(n\alpha(n))$ points.
		\item The lower envelope of $S_1$ and the upper envelope of $S_2$ (and vice versa) intersect one another in at most $O(n)$ points.
		\item If no two segments in $S_1$ cross, then each envelope of $S_1$ intersects each envelope of $S_2$ in at most $O(n)$ points.
		\item If every segment in $S_1$ has larger slope than every segment in $S_2$, then each envelope of $S_1$ intersects each envelope of $S_2$ in at most $O(n)$ points.
	\end{enumerate}
	Furthermore, all these bounds are worst-case tight.
\end{lemma}

\begin{proof}
	Claim (1) follows from the Hart--Sharir result applied on $S_1\cup S_2$. Worst-case tightness follows by taking a configuration $S$ of $n$ segments that realizes $\Theta(n\alpha(n))$ lower-envelope intersections, and considering a random partition of $S$ into two parts. Each lower-envelope intersection has a probability of $1/2$ of its two segments falling into different parts. By linearity of expectation, the expected number of lower-envelope intersections between segments of different parts is $\Theta(n\alpha(n)/2)=\Theta(n\alpha(n))$, and there must be a partition that realizes at least this value.
	
	For Claims (2) and (3), partition the plane into $2n+1$ vertical slabs by tracing vertical lines through the endpoints of the segments. Within each slab, the lower envelope of each $S_i$ is part of a single concave chain, and the upper envelope of each $S_i$ is part of a single convex chain. Then, Claim (2) follows since a concave chain and convex chain can intersect in at most two points. And Claim (3) follows since in this case, each chain of $S_1$ consists of a single segment, so it intersects a chain of $S_2$ in at most two points.
	
	For the cases of Claim (4) not yet proven, consider the lower envelope of $S_1\cup S_2$. Due to the constraints on the slopes, each concave chain in the envelope contains at most one intersection between a segment of $S_1$ and a segment of $S_2$. The case of the upper envelope is symmetric.
\end{proof}

\section{Triple vertex contacts with a triangular robot}\label{sec_triangular}

Following Halperin and Yap \cite{HY}, we first consider the following auxiliary problem: Suppose the robot $B$ is a ``flat'' triangle that can translate amidst polyhedral obstacles in $\R^3$. We would like to bound the number of vertices of the free space $\mathcal F$ that correspond to triple vertex contacts, i.e., the number of free placements in which the three vertices of $B$ make contact with three different obstacle faces (see Figure \ref{fig_triangular_robot}). Halperin and Yap derived an upper bound of $O(n^2\alpha(n))$ for the number of such contacts. We improve this bound to $O(n^2)$ by making a minor modification to their argument.

\begin{lemma}\label{lem_triangular_vvv}
	Let $B$ be a triangular robot that is free to translate in $\R^3$ amidst polyhedral obstacles of total complexity $n$. Then the number of triple vertex contacts that $B$ can make with obstacles is $O(n^2)$.
\end{lemma}

(Recall that if we do not restrict our attention to vertex contacts then $\Omega(n^2\alpha(n))$ triple contacts are possible, as shown in Figure \ref{fig_triangle_lower_bound}.)

\begin{figure}
	\centering\includegraphics[width=7cm]{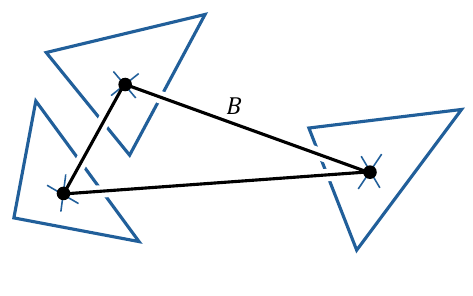}
	\caption{\label{fig_triangular_robot}A ``flat'' triangular robot $B$ in $\R^3$ making three vertex contacts.}
\end{figure}

The general idea of the proof is as follows: For every contact specification $O_1=(p_1,f_1)$ where $p_1$ is a robot vertex and $f_1$ is an obstacle face, we represent the set robot placements that realize $O_1$ in a parametric plane $P_{O_1}$, in which the two coordinates represent the two degrees of freedom the robot has in realizing $O_1$. Then, each second contact $O_2=(p_2,f_2)$ involving a second vertex of $B$ corresponds to a straight-line segment within $P_{O_1}$. For each choice of $O_1$ and $p_2$, there is a family of such segments, and these segments are pairwise disjoint. Each triple contact $O_1$, $O_2$, $O_3$ involving the three vertices of $B$ corresponds to the intersection between two segments in $P_{O_1}$ belonging to two such segment families. Under the right circumstances, the intersection points corresponding to triple contacts will belong to the upper or lower envelope of each segment family, so Lemma \ref{lem_env} can be applied, specifically Case 3 of the lemma. In other cases, a bit more work will be required.

Let us now proceed with the actual proof of Lemma \ref{lem_triangular_vvv}.

Suppose without loss of generality that the triangle $B$ is ``horizontal'', i.e., parallel to the $xy$-plane. We can assume all obstacle faces are triangles. By the general-position assumption, the vertices of each triangular obstacle face have different $z$-coordinates. Let us split each face into two \emph{subtriangles} by a horizontal line passing through the vertex with middle $z$-coordinate. Thus, each subtriangle has two non-horizontal sides and one horizontal side.

Let us see what happens to the intersection $f(z_0)$ between a subtriangle $f$ and a horizontal plane $z=z_0$, as we vary $z_0$. For every $z_0$ in the $z$-range of $f$, $f(z_0)$ is a line segment. As we increase $z_0$ at a constant rate, the line segment moves parallel to itself at constant speed, and its two endpoints also move at constant speed.

\subsection{Contact specifications and covering}

In our case, the relevant contact specifications are those of the form $O=(p,f)$ where $p$ is a vertex of $B$ and $f$ is a subtriangle.

\begin{figure}
	\centering\includegraphics[width=14cm]{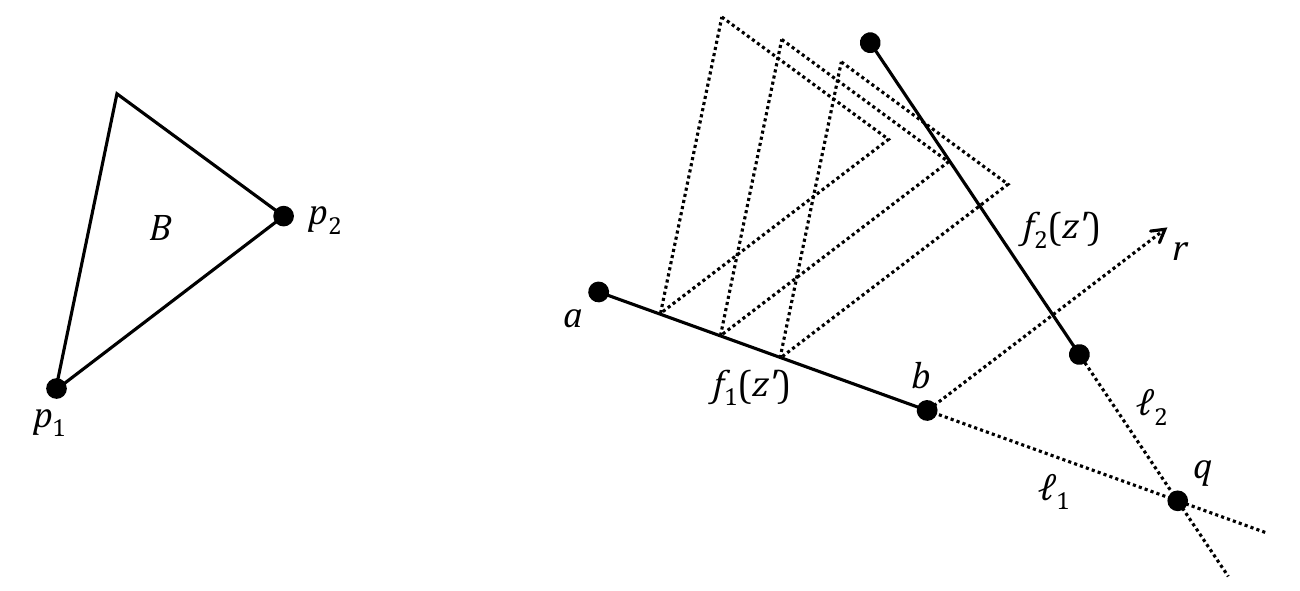}
	\caption{\label{fig_covers}At this value $z=z'$, the contact specification $O_2=(p_2,f_2)$ covers the contact specification $O_1=(p_1,f_1)$ at the right, since $q$ is to the right of $f_1(z')$ and the ray $r$ from $b$ with direction $\overrightarrow{p_1p_2}$ intersects $f_2(z')$. As we slide the robot $B$ along $f_1(z')$, it goes from not colliding with $f_2(z')$ to colliding with it.}
\end{figure}

Consider two contact specifications $O_1=(p_1, f_1)$, $O_2=(p_2, f_2)$, where $p_1\neq p_2$, and where the $z$-ranges of $f_1$, $f_2$ overlap. Let $[z_0, z_1]$ be the intersection of the $z$-ranges of $f_1$, $f_2$. Given $z'\in[z_0, z_1]$, we would like to define whether $O_2$ \emph{covers} $O_1$ at $z=z'$. Let $\ell_1=\ell_1(z')$ and $\ell_2=\ell_2(z')$ be the lines containing the segments $f_1(z')$ and $f_2(z')$. Let $q$ be the point of intersection of $\ell_1$ and $\ell_2$, and suppose that $f_1(z')$ does not contain $q$. Let $a$, $b$ be the endpoints of $f_1(z')$, with $a$ left of $b$ (i.e., with smaller $x$-coordinate). Let $r$ be the ray with direction $\overrightarrow{p_1p_2}$ emerging from the endpoint among $a$, $b$ that is closer to $q$. If $r$ intersects $f_2(z')$ then we say that $O_2$ \emph{covers} $O_1$. If $q$ is to the right of $f_1(z')$ then we say that $O_2$ covers $O_1$ \emph{at the right}; otherwise we say that $O_2$ covers $O_1$ \emph{at the left}. See Figure \ref{fig_covers}.

\begin{figure}
	\centering\includegraphics[width=13cm]{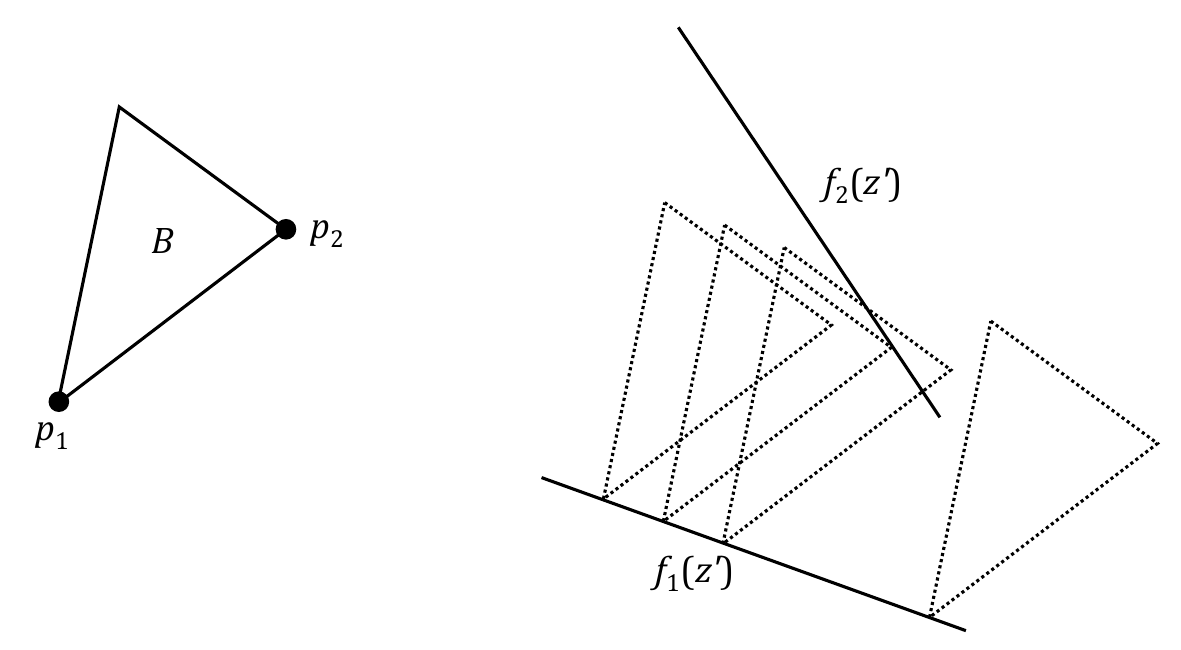}
	\caption{\label{fig_does_not_cover}Here $O_2=(p_2,f_2)$ does not cover $O_1=(p_1,f_1)$, and as we slide the robot $B$ along $f_1(z')$, it collides with $f_2(z')$ but then stops colliding with it.}
\end{figure}

The significance of $O_2$ covering $O_1$ is as follows. Suppose we slide the robot with its vertex $p_1$ moving along the segment $f_1(z')$ from one endpoint of $f_1(z')$ to the other. It might happen that during this motion, the edge $p_1p_2$ starts colliding with the segment $f_2(z')$ and then stops colliding with it. See Figure \ref{fig_does_not_cover}. However, if $O_2$ covers $O_1$ then this is not possible: If the edge $p_1p_2$ starts colliding with the segment $f_2(z')$, then the collision will continue until $p_1p_2$ reaches the other endpoint of $f_1(z')$.

For each generic value $z'\in[z_0,z_1]$, either $O_1$ covers $O_2$ or $O_2$ covers $O_1$. Furthermore, since the endpoints of the segments $f_1(z')$, $f_2(z')$ move linearly with respect to $z'$, the range of values of $z'$ at which $O_1$ covers $O_2$ forms a contiguous subinterval of $[z_0,z_1]$, and the range of values of $z'$ at which $O_2$ covers $O_1$ forms another contiguous subinterval of $[z_0,z_1]$.

Consider a placement of the robot realizing three vertex contact specifications $O_1=(p_1, f_1)$, $O_2=(p_2, f_2)$, $O_3=(p_3, f_3)$. Let $z'$ be the $z$-coordinate of this robot placement. For each pair of indices $i,j\in\{1,2,3\}$, $i\neq j$, either $O_i$ covers $O_j$ or $O_j$ covers $O_i$ at $z'$. Hence, one of the two following possibilities occurs: Either there is one contact specification that is covered by the other two, or the contact specifications are \emph{circular} in the sense that $O_i$ covers $O_j$, $O_j$ covers $O_k$, and $O_k$ covers $O_i$. See Figure \ref{fig_3_vertex_contacts}.

\begin{figure}
	\centering\includegraphics[width=14cm]{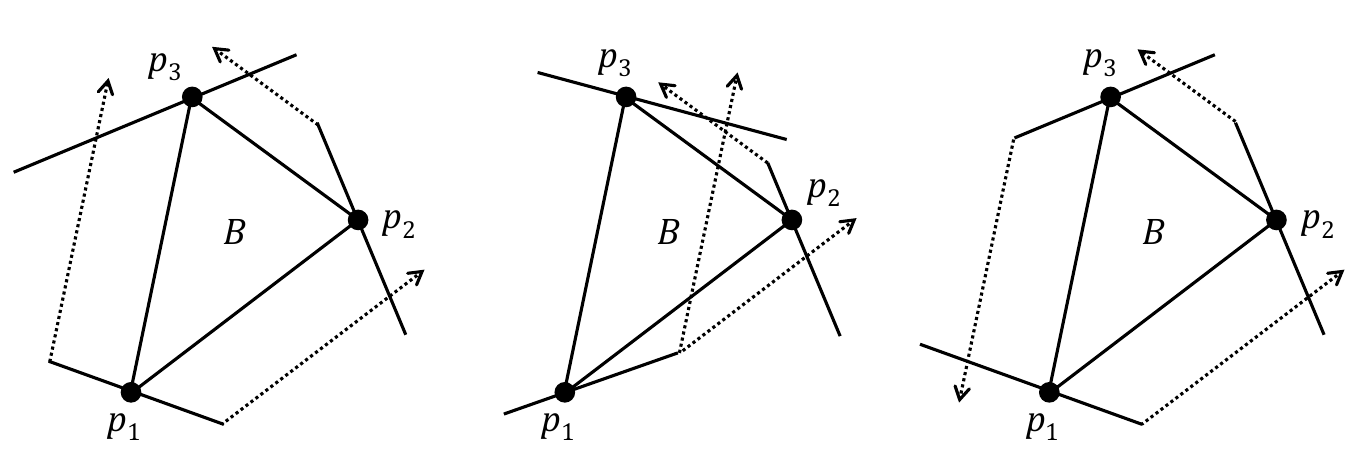}
	\caption{\label{fig_3_vertex_contacts}Robot placements realizing three vertex contact specifications $O_1$, $O_2$, $O_3$. In the left and center figures, $O_1$ is covered by $O_2$ and $O_3$. In the right figure, $O_1$ is covered by $O_2$, which is covered by $O_3$, which is covered by $O_1$.}
\end{figure}

\subsection{Parametric plane of a contact specification}

Given a contact specification $O_1=(p_1,f_1)$, we define the \emph{contact parametric plane} $P_{O_1}$ in which we parametrize by a pair of real numbers $(z', s)$ every robot placement in which the robot vertex $p_1$ makes contact with $f_1(z')$ (or with its supporting line). The number $s$ represents the distance between the left endpoint of $f_1(z')$ and the position of $p_1$ along $f_1(z')$. Thus, $P_{O_1}\cong \R^2$, and the robot placements that realize the contact $O_1$ correspond to a triangular region $\Delta_{O_1}\subset P_{O_1}$. We imagine that the $z$-axis in $P_{O_1}$ is horizontal and the $s$-axis is vertical.

We can mark each point of $\Delta_{O_1}$ as \emph{legal} or \emph{illegal} depending on whether the corresponding robot placement intersects the interior of another obstacle or not.

Given another contact specification $O_2=(p_2,f_2)$ with $p_2\neq p_1$, the set of robot placements that realize both $O_1$ and $O_2$ corresponds either to the empty set or to a line segment within $\Delta_{O_1}$. However, we plot in $P_{O_1}$ only contact specifications that cover $O_1$, and only for the $z$-ranges at which they cover it. Given $O_2$ that covers $O_1$ at a certain $z$-range, let $\sigma_{O_2}$ be the segment in $\Delta_{O_1}$ that corresponds to all robot placements that realize the double contact $O_1$, $O_2$ in this $z$-range. If $O_2$ covers $O_1$ at the left, then all the points vertically below $\sigma_{O_2}$ in the triangle $\Delta_{O_1}$ are illegal, and if $O_2$ covers $O_1$ at the right, then all the points vertically above $\sigma_{O_2}$ in $\Delta_{O_1}$ are illegal. See Figure~\ref{fig_parametric_plane}.

\begin{figure}
	\centering\includegraphics[width=16cm]{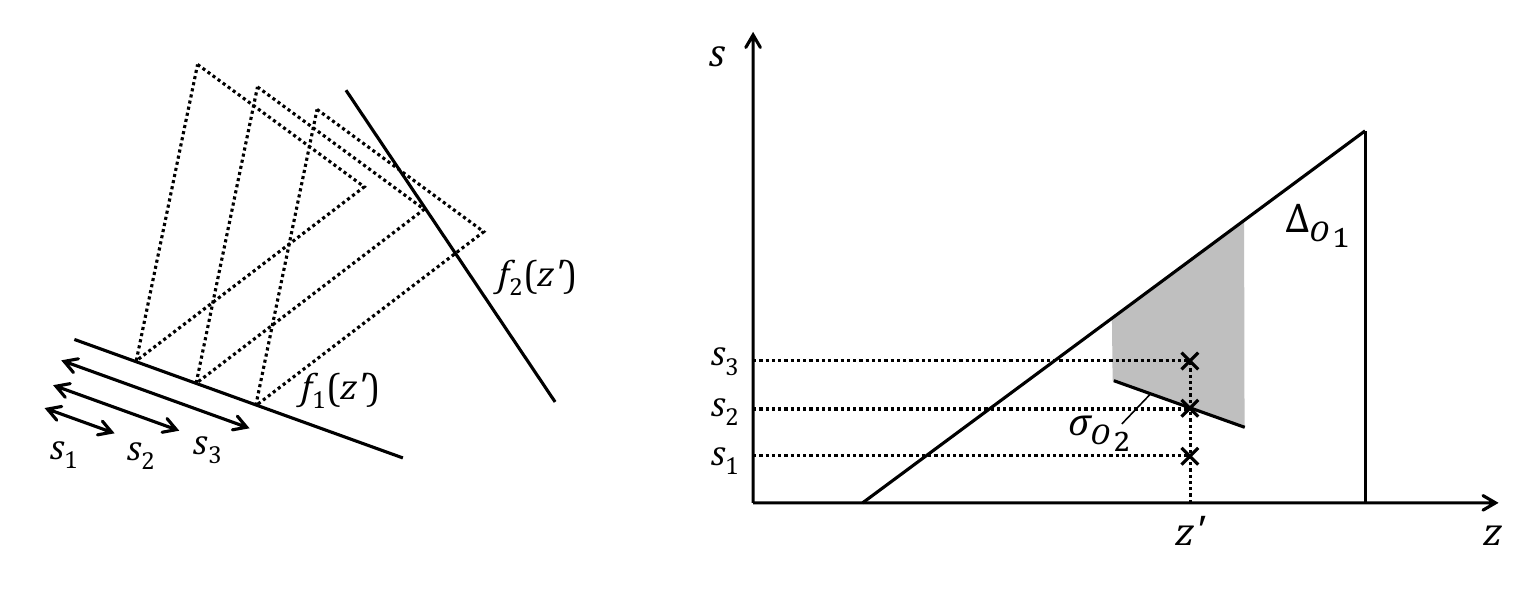}
	\caption{\label{fig_parametric_plane}Three robot placements realizing a contact specification $O_1$ with the same $z$-value (left); and their representation in the parametric plane $P_{O_1}$ (right). The shaded area within $\Delta_{O_1}$ represents robot placements that are illegal due to collision with $f_2$.}
\end{figure}

\subsection{Families of non-intersecting segments in the parametric planes}

Given a contact specification $O_1=(p_1,f_1)$ and another robot vertex $p_2$, let $L_{O_1,p_2}$ (resp.~$R_{O_1,p_2}$) be the set of all subtriangles $f$ such that $O_2=(p_2,f)$ covers $O_1$ at the left (resp.~right) for some range of $z$. For each $f\in L_{O_1,p_2}$ (resp.~$R_{O_1,p_2}$) let $\sigma(f)$ be the corresponding line segment in the parametric plane $P_{O_1}$. Let $\mathcal L_{O_1,p_2} = \{\sigma(f):f\in L_{O_1,p_2}\}$ and $\mathcal R_{O_1,p_2} = \{\sigma(f):f\in R_{O_1,p_2}\}$ be the sets of these line segments.

Here is the crucial observation that is missing in \cite{HY}: The segments in $\mathcal L_{O_1,p_2}$ are pairwise non-intersecting, since if $\sigma(f), \sigma(f')\in \mathcal L_{O_1,p_2}$ intersected, then the subtriangles $f$, $f'$ themselves would intersect. Similarly, the segments in  $\mathcal R_{O_1,p_2}$ are pairwise non-intersecting.

\subsection{Counting triple vertex contacts}\label{subsec_counting}

We are now ready to bound the number of robot placements making three simultaneous vertex contacts.

Consider first triple contacts in which one contact specification $O_1=(p_1,f_1)$ is covered by the other two, $O_2=(p_2,f_2)$, $O_3=(p_3,f_3)$. In the parametric plane $P_{O_1}$, such a triple contact corresponds to the intersection point $q$ of a segment in either $\mathcal L_{O_1,p_2}$ or $\mathcal R_{O_1,p_2}$ with a segment in either $\mathcal L_{O_1,p_3}$ or $\mathcal R_{O_1,p_3}$. Furthermore, in order for $q$ to be legal, it is necessary for $q$ to lie in the upper (resp.~lower) envelope of $\mathcal L_{O_1, p_2}$ (resp.~$\mathcal R_{O_1,p_2}$), as well as in the upper (resp.~lower) envelope of $\mathcal L_{O_1, p_3}$ (resp.~$\mathcal R_{O_1,p_3}$). Hence, Lemma \ref{lem_env} (Case 3) bounds the number of these intersection points $q$, for fixed $O_1$, by $O(n)$. Since there are $O(n)$ choices of $O_1$, the total bound for this type of triple-intersection contacts is $O(n^2)$.

Now consider circular triple contacts. Consider a robot placement at $z=z'$ realizing the triple contact $O_1=(p_1,f_1)$, $O_2=(p_2,f_2)$, $O_3=(p_3,f_3)$ in which $O_1$ is covered by $O_2$, $O_2$ is covered by $O_3$, and $O_3$ is covered by $O_1$. Say that all the coverings are at the right (the other cases are similar). This robot placement corresponds to a point $q_1$ in the segment $\sigma(f_2)\in \mathcal R_{O_1,p_2}$, which is the lowest segment of $\mathcal R_{O_1,p_2}$ at $z=z'$. The robot placement also corresponds to a point $q_2$ in the segment $\sigma(f_3)\in\mathcal R_{O_2,p_3}$, which is the lowest segment of $\mathcal R_{O_2,p_3}$. Finally, the robot placement also corresponds to a point $q_3$ in the segment $\sigma(f_1)\in\mathcal R_{O_3,p_1}$, which is the lowest segment of $\mathcal R_{O_3,p_1}$. All three points $q_1$, $q_2$, $q_3$ have $z$-coordinate equal to $z'$.

Consider the lower envelope of $\mathcal R_{O_1,p_2}$. This lower envelope is piecewise linear, with breakpoints at discrete values of $z$. Since the segments in $\mathcal R_{O_1,p_2}$ are pairwise nonintersecting, each breakpoint is an endpoint of a segment of $\mathcal R_{O_1,p_2}$. Let $z_1$ be the largest $z$-coordinate of a breakpoint that is smaller than $z'$. Define $z_2$, $z_3$ similarly, by looking at $\mathcal R_{O_2,p_3}$, $\mathcal R_{O_3,p_1}$, respectively. Let $z'' = \max\{z_1,z_2,z_3\}$.

Hence, in the $z$-range $[z'',z']$ the lowest segment of each $\mathcal R_{O_1,p_2}$, $\mathcal R_{O_2,p_3}$, $\mathcal R_{O_3,p_1}$ is the aforementioned $\sigma(f_2)$, $\sigma(f_3)$, $\sigma(f_1)$, respectively. We charge the triple contact $O_1$, $O_2$, $O_3$ to the breakpoint $z''$.

In this charging scheme, each breakpoint in a lower envelope of a parametric plane is charged at most once, because we can reconstruct the triple contact given the breakpoint: Given $z''$ which is a breakpoint in $\mathcal R_{(p_1,f_1),p_2}$, say, we take the segment $\sigma(f_2)$ that is lowest in $\mathcal R_{(p_1,f_1),p_2}$ just after $z=z''$. Once we know $f_2$, we know we need to look at $\mathcal R_{(p_2,f_2),p_3}$. Then we take the segment $\sigma(f_3)$ that is lowest in $\mathcal R_{(p_2,f_2),p_3}$ just after $z=z''$. Once we know $f_3$, we know we need to look at $\mathcal R_{(p_3,f_3),p_1}$. Then we take the segment $\sigma(f_4)$ that is lowest in $\mathcal R_{(p_3,f_3),p_1}$ just after $z=z''$. If $f_4=f_1$ then we have successfully reconstructed the triple contact-specification $(p_1,f_1)$, $(p_2,f_2)$, $(p_3,f_3)$ (which might or might not be realizable as a triple contact). If $f_4\neq f_1$ then the breakpoint $z''$ we started with does not receive any charge.

Since each breakpoint is a segment endpoint and there are $O(n^2)$ segments in all the families $\mathcal R_{O,p}$ altogether, there are at most $O(n^2)$ circular triple contacts in which all coverings take place at the right.

The other possibilities, in which some of the coverings take place at the left, are treated similarly: Whenever a covering takes place at the left, the corresponding segment is the highest one in the corresponding family $\mathcal L_{O,p}$, instead of the lowest one in $\mathcal R_{O,p}$.

This concludes the proof of Lemma \ref{lem_triangular_vvv}.

\begin{remark}
	Halperin and Yap \cite{HY} defined each $\mathcal L_{(p_1,f_1),p_2}\cup\mathcal L_{(p_1,f_1),p_3}$ as a single set, instead of as two separate sets. Therefore, the segments in their sets can intersect, and that is why they got the $\alpha(n)$ factor in their bound. Our approach of defining $\mathcal L_{(p_1,f_1),p_2}$ and $\mathcal L_{(p_1,f_1),p_3}$ as two separate sets (and similarly defining $\mathcal R_{(p_1,f_1),p_2}$ and $\mathcal R_{(p_1,f_1),p_3}$ as two separate sets) allows us to apply Lemma \ref{lem_env} (Case 3) and get a bound free of the $\alpha(n)$ factor.
\end{remark}

Following Halperin and Yap, we conclude the following:

\begin{corollary}\label{cor_poly_vvv}
	Let $B$ be a polyhedral robot of constant complexity that is free to translate in $\R^3$ amidst polyhedral obstacles of total complexity $n$. Then the number of triple vertex contacts that $B$ can make with obstacles is $O(n^2)$.
\end{corollary}

\begin{proof}
	Every triple contact made by three vertices of $B$ is also a triple contact made by a triangular robot $B'$ spanned by those three vertices (though the opposite is not necessarily the case).
\end{proof}

\section{The case of a cubical robot}\label{sec_cube}

In this section we prove Theorem \ref{thm_box}, regarding the case in which $B$ is a cube (or a box or a parallelepiped).

The technique is similar to that of the previous section. We identify triple contacts with intersections between segments in parametric planes. We show that in each case, the intersections must belong to the lower or upper envelopes of the families containing the segments. And we show that in each case we can apply Case 2, 3, or 4 of Lemma \ref{lem_env}, so the unfortunate Case 1 is always avoided.

Assume for concreteness that $B=[0,1]^3$ is an axis-parallel cube of side-length $1$.Let us go over all types of triple contacts $O_1=(f_1, g_1)$, $O_2=(f_2, g_2)$, $O_3=(f_3,g_3)$.

Lemma \ref{lem_triangular_vvv} already addressed the case of three vertex contacts. Hence, we can assume without loss of generality that $O_1$ is not a vertex contact.

Due to the general position assumption, it is impossible to have two face contacts involving the same face or opposite faces of $B$.

\begin{lemma}\label{lem_axis_parallel_movement}
	In each of the following cases there are at most $O(n^2)$ triple contacts involving $f_1$ and $f_2$:
	\begin{itemize}
	\item $f_1$ and $f_2$ are the same edge or parallel edges of $B$.
	\item $f_1$ and $f_2$ are nonparallel faces of $B$.
	\item $f_1$ is a face of $B$, and $f_2$ is an edge of $f_1$ or of the face opposite $f_1$.
	\end{itemize}
\end{lemma}

\begin{proof}
	In all these cases, once we fix obstacle features $g_1$, $g_2$ for the first two contacts, the line of movement of $B$ that preserves these two contacts is axis-parallel. Suppose we move $B$ in one direction along this line until it makes a third contact with some obstacle feature $g_3$. Further movement in the same direction will make $B$ intersect $g_3$ until $B$ moves distance at least $1$. But after $B$ moves such a distance, it will no longer realize $O_1$ nor $O_2$. Hence, for fixed $g_1$, $g_2$, there are at most two possibilities for $g_3$, one for each direction of movement. Since there are $O(n^2)$ choices for $g_1$, $g_2$, the claim follows. (This argument is given in Lemma 2.1 of \cite{HY}.)
\end{proof}

In order to handle the remaining cases, we examine parametric planes of face and edge contacts.

\subsection{Parametric plane of a face contact}

Let $O_1=(f_1, g_1)$ be a face contact. The set of all placements of $B$ that realize $O_1$ can be parametrized by a pair of real numbers between $0$ and $1$. For example, if the face $f_1$ is parallel to the $xy$-plane, then we can represent every placement of $B$ that realizes $(f_1,g_1)$ by two real numbers $s,t\in[0,1]$, where $s$ denotes the distance from the vertex $g_1$ to the higher-$x$, $y$-parallel edge of $f_1$, and $t$ denotes the distance from $g_1$ to the higher-$y$, $x$-parallel edge of $f_1$. Real values of $s$ and/or $t$ outside of the range $[0,1]$ give rise to placements in which $g_1$ lies in the containing plane of $f_1$, though not in $f_1$ itself.

In the parametric plane $P_{O_1}\cong \R^2$, the unit square $\Gamma=[0,1]^2\subset P_{O_1}$ corresponds to the placements in which the contact $O_1$ actually takes place. Each point of $P_{O_1}$ can be marked \emph{legal} or \emph{illegal} according to whether the corresponding robot placement intersects the interior of some obstacle or not.

As before, given another contact specification $O_2=(f_2,g_2)$, the set of placements that realize both $O_1$ and $O_2$ corresponds either to a line segment $\sigma_{O_1,O_2} \subset P_{O_1}$ or to the empty set. For each choice of $f_2$, let $S_{O_1,f_2}$ be the set of all segments $\sigma_{O_1,(f_2,g_2)}\subset P_{O_1}$ over all valid choices of $g_2$.

\begin{figure}
	\centering\includegraphics[width=12cm]{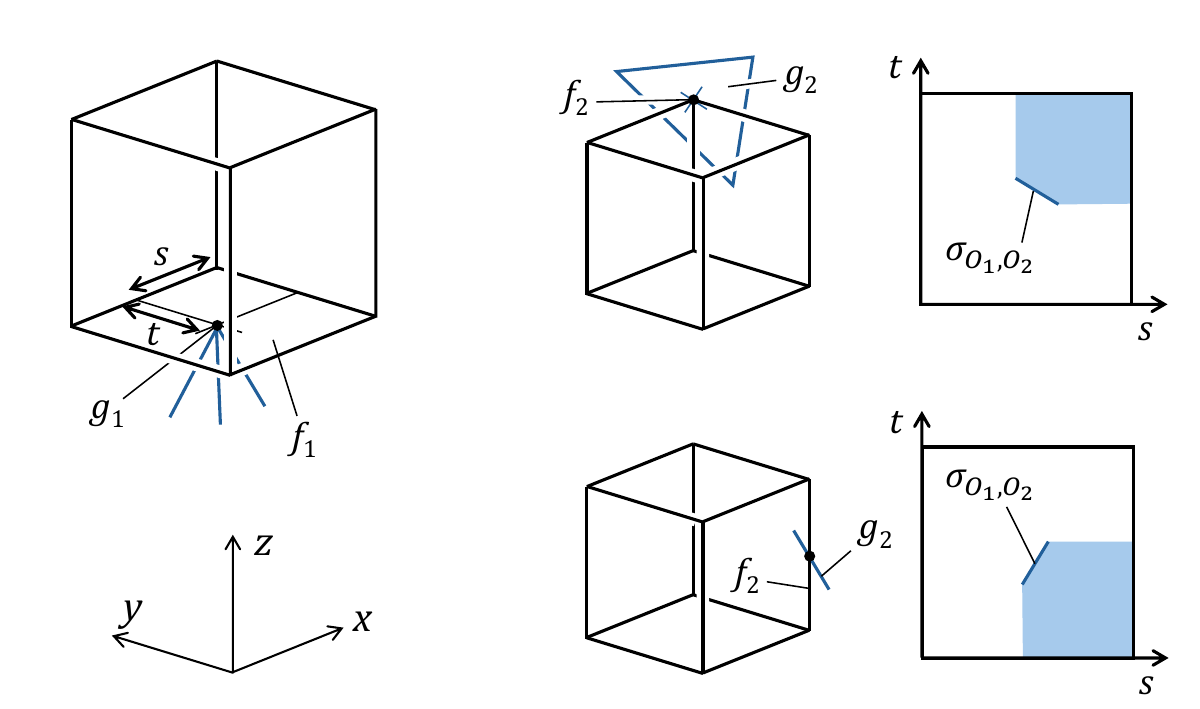}
	\caption{\label{fig_cube_facePP}Parametric plane of a face contact for a cube robot. The contact $O_1=(f_1,g_1)$ defining the parametric plane is shown on the left, and two different options for the second contact $O_2=(f_2, g_2)$, and the corresponding segments $\sigma_{O_1,O_2}$ in the parametric plane, are shown on the right.}
\end{figure}

Together with the line segment $\sigma_{O_1,O_2}$, the contact specification $O_2$ also induces an illegal region in $\Gamma$. Suppose, as before, that the face $f_1$ is parallel to the $xy$-plane. Recall that we are only interested in cases where $f_2$ is either a vertex of $B$, or an edge of $B$ perpendicular to $f_1$. Suppose that $f_2$ is either one of the top-left vertices of $B$ (after projecting to the $xy$-plane) or the edge of $B$ connecting these two vertices. Then the region of $\Gamma$ that lies above and to the left of the segment $\sigma_{O_1,O_2}$ is illegal. Similarly for the other cases. For each choice of $f_2$, the region of $\Gamma$ that lies behind $\sigma_{O_1,O_2}$ in two perpendicular axis-parallel directions is illegal. See Figure \ref{fig_cube_facePP}. Furthermore, if $f_2$ is a vertex of $B$, then the segments in $S_{O_1,f_2}$ are pairwise non-intersecting, because an intersection between two segments $\sigma_{O_1,(f_2,g_2)}$, $\sigma_{O_1,(f_2,g'_2)}$ would yield an intersection between the corresponding obstacle faces $g_2$, $g_2'$.

\subsection{Parametric plane of an edge contact}

\begin{figure}
	\centering\includegraphics[width=15cm]{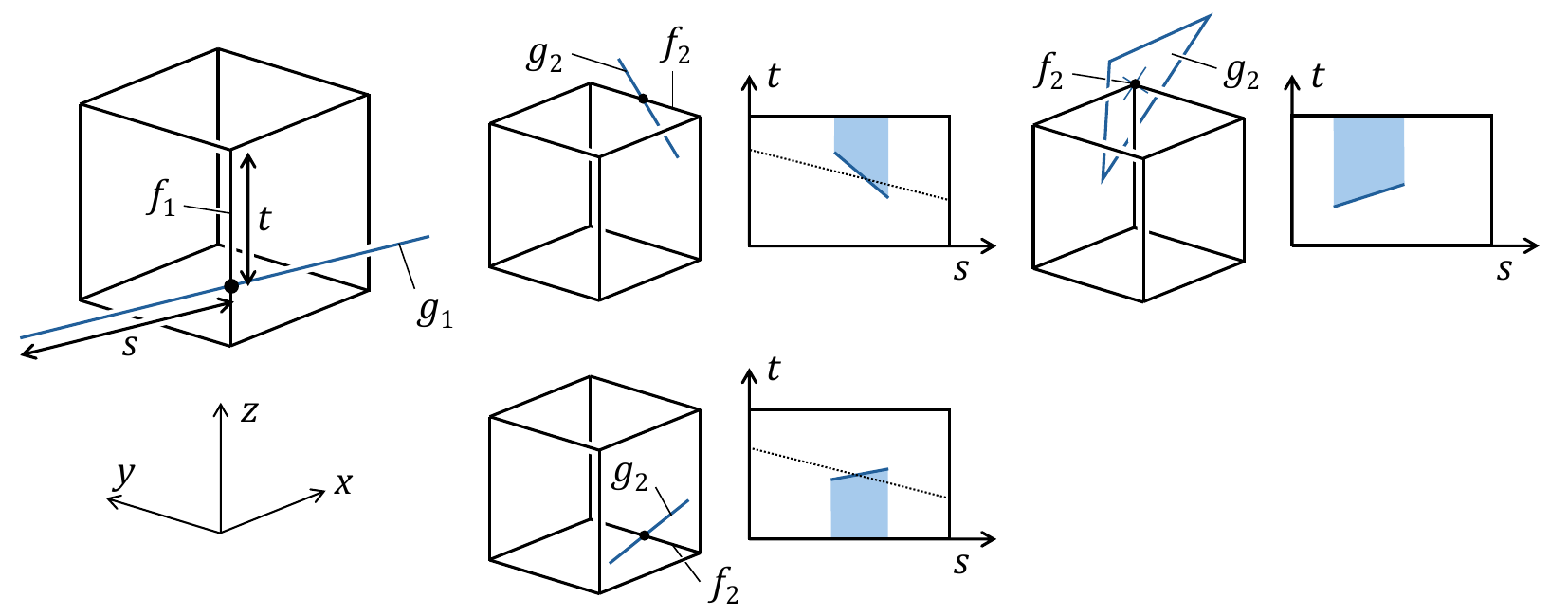}
	\caption{\label{fig_cube_edgePP_new}Parametric plane of an edge contact for a cube robot. The dotted lines in the parametric planes represent robot movements that preserve $O_1$ and are ``horizontal'' (parallel to the $xy$-plane).}
\end{figure}

Now let $O_1=(f_1, g_1)$ be an edge contact. Assume for concreteness that $f_1$ is the ``vertical'' edge of $B$ that satisfies $x=y=0$. In order for the contact $O_1$ to be possible, the obstacle edge $g_1$ must decrease in $y$-coordinate as it increases in $x$-coordinate.

We can represent each placement of $B$ in which $f_1$ (or its extension) intersects $g_1$ (or its extension) by a pair $(s, t)$, where $s$ is the distance between the lower-$x$ endpoint of $g_1$ and the point of contact $p$, and $t$ is the distance between $p$ and the higher-$z$ endpoint of $f_1$.

In the resulting parametric plane $P_{O_1}\cong \R^2$, the placements in which the contact $O_1$ actually takes place correspond to a rectangle $\Gamma=[0, s_{\mathrm{max}}]\times [0,1]\subset P_{O_1}$, where $s_{\mathrm{max}}$ is the length of $g_1$. A movement of $B$ that preserves the contact $O_1$ and is ``horizontal'' (preserves $z$-coordinate) corresponds in $P_{O_1}$ to a line of a certain slope $m=m_{g_1}$.

Each point of $P_{O_1}$ can be marked legal or illegal according to whether the corresponding placement of $B$ is legal or illegal. Consider a second contact specification $O_2=(f_2, g_2)$, where $f_2$ is either a vertex of $B$ or an edge of $B$ not parallel to $f_1$. Then the set of points corresponding to placements that realize $O_2$ is either empty or a line segment $\sigma_{O_1,O_2}\subset P_{O_1}$.  For each such choice of $f_2$, let $S_{O_1,f_2}$ be the set of all segments $\sigma_{O_1,(f_2,g_2)}\subset P_{O_1}$ over all valid choices of $g_2$. Furthermore, the region of $\Gamma$ that is either vertically above or vertically below the segment $\sigma_{O_1,O_2}$ (depending on whether $f_2$ lies on the lower-$z$ or the higher-$z$ side of $B$) is illegal. See Figure \ref{fig_cube_edgePP_new}.

As before, if $f_2$ is a vertex of $B$, then the segments in $S_{O_1,f_2}$ are pairwise non-intersecting.

\begin{figure}
	\centering\includegraphics[width=6cm]{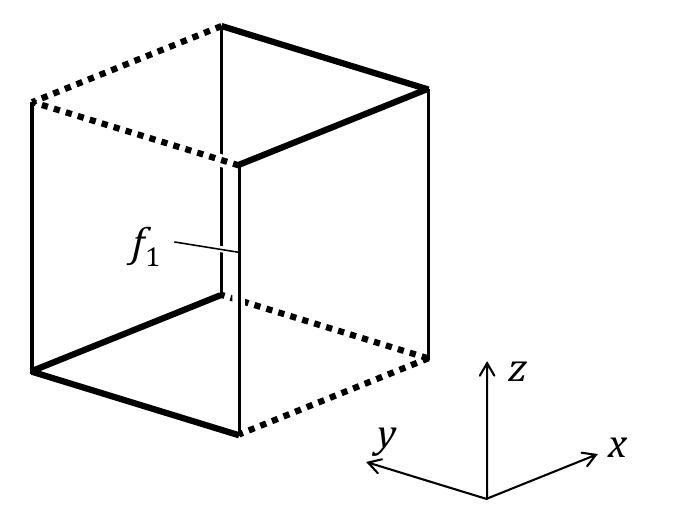}
	\caption{\label{fig_compare_to_m}If $f_2$ is one of the bolded edges, then a robot movement that preserves $O_1=(f_1,g_1)$ and $O_2=(f_2,g_2)$ will descend (decrease in $z$-coordinate) as it increases in $x$-coordinate. If $f_2$ is one of the dotted edges, then a robot movement that preserves $O_1$ and $O_2$ will ascend as it increases in $x$-coordinate.}
\end{figure}

Consider an edge $f_2$ of $B$ that is not parallel to $f_1$. For example, let $f_2$ be the edge with $x=z=1$, which is parallel to the $y$-axis. In order for a contact $O_2=(f_2,g_2)$ to be possible, the obstacle edge $g_2$ must decrease in $z$-coordinate as it increases in $x$-coordinate. Hence, a motion that simultaneously preserves $O_1$ and $O_2$ must decrease in $z$-coordinate as $s$ increases, so the corresponding segment $\sigma_{O_1,O_2}$ will have slope smaller than $m$.

In general, out of the eight edges $f_2$ of $B$ that are not parallel to $f_1$, four of them yield segments $\sigma_{O_1,O_2}$ that have slope larger than $m$ for all valid choices of obstacle edge $g_2$. The other four edges $f_2$ of $B$ yield segments $\sigma_{O_1,O_2}$ that have slope smaller than $m$ for all valid choices of obstacle edge $g_2$. See Figure \ref{fig_compare_to_m}.

\subsection{The case analysis}

Let us now go over the possibilities not handled by Lemmas \ref{lem_triangular_vvv} and \ref{lem_axis_parallel_movement}.

Suppose at least one contact is a face contact. Say $f_1$ is a face of $B$. Then $f_2$, $f_3$ cannot both be edges of $B$, because then either they would be parallel, or else one of them would belong to $f_1$ or to the face opposite $f_1$. Thus, one of $f_2$, $f_3$, say $f_2$, must be a vertex of $B$, while $f_3$ must be either another vertex of $B$ or an edge perpendicular to $f_1$. Fix $g_1$. Each triple contact involving $O_1,f_2,f_3$ corresponds to an intersection between a segment of $S_{O_1,f_2}$ and a segment of $S_{O_1,f_3}$, that lies in an envelope (upper or lower) of each of these segment sets. Since the segments in $S_{O_1,f_2}$ are pairwise non-intersecting, Lemma \ref{lem_env} (Case 3) applies. There are $O(n)$ triple contacts for our choice of $g_1$, and hence a total of $O(n^2)$ contacts involving our choice of $f_1$, $f_2$, $f_3$.

Now suppose no contact is a face contact. At least one contact, say $O_1$, must be an edge contact. Hence, let $f_1$ be an edge of $B$. Each of $f_2$, $f_3$ must be an edge or a vertex of $B$, and no two of $f_1$, $f_2$, $f_3$ can be parallel edges. If one of $f_2$, $f_3$ is a vertex of $B$, then Lemma \ref{lem_env} (Case 3) applies again, and we again get a bound of $O(n^2)$.

Now suppose $f_1$, $f_2$, $f_3$ are three mutually nonparallel edges of $B$. For each $f_i$, let $h_{i,1}$, $h_{i,2}$ be the two faces of $B$ that are perpendicular to $f_i$. If at least one $f_i$ satisfies the property that the other two edges $f_j$, $f_k$ belong to different faces $h_{i,1}$, $h_{i,2}$, then in the parametric plane $P_{O_i}$ we are looking at intersections between an upper envelope and a lower envelope, so Lemma \ref{lem_env} (Case 2) applies, and we get once again a bound of $O(n^2)$.

The only remaining case is where $f_1$, $f_2$, $f_3$ are three edges incident to the same vertex of $B$. Fix an obstacle edge $g_1$. In the parametric plane $P_{O_1}$, every segment in $S_{O_1,f_2}$ has slope smaller than some slope $m=m_{g_1}$, while every segment in $S_{O_1,f_3}$ has slope larger than $m$. Hence, Lemma \ref{lem_env} (Case 4) applies, and we again get a bound of $O(n)$ triple contacts for our choice of $g_1$, or $O(n^2)$ contacts in total.

This case analysis concludes the proof of Theorem \ref{thm_box}.

\section{The case of a fully-parallel polygonal robot}

\begin{proof}[Proof of Theorem \ref{thm_fully-parallel}]

Let $B$ be a fully-parallel polygonal robot. We proceed as in the proof of Theorem \ref{thm_box}, but with a few small modifications. Every pair of parallel edges $f_1$, $f'_1$ of $B$ partition the rest of the boundary of $B$ into two parts. We denote these two parts $\pi_1(f_1)$ and $\pi_2(f_1)$. Each endpoint of $f_1$, $f'_1$ belongs to one of these two parts as well.

Let us go through all possible triple contacts $O_1=(f_1, g_1)$, $O_2=(f_2,g_2)$, $O_3=(f_3,g_3)$.

First, consider the case in which $O_1=(f_1,g_1)$ is a face contact, meaning $f_1$ is $B$ itself and $g_1$ is an obstacle vertex. By the general position assumption, no other contact can be a face contact. In order to realize $O_1$, the robot $B$ is restricted to move in a plane parallel to $B$. As mentioned in the Introduction, in the planar case the complexity of the free space is $O(n)$. Since there are $O(n)$ choices of $O_1$, we get at most $O(n^2)$ triple contacts of this type.

Now consider the case where two edge contacts $O_1=(f_1, g_1)$, $O_2=(f_2, g_2)$ involve robot edges $f_1$, $f_2$ that are either equal or parallel. Let $d$ and $D$, with $d<D$, be the lengths of $f_1$ and $f_2$. Here the argument is similar to that of Lemma \ref{lem_axis_parallel_movement}. The contacts $O_1$, $O_2$ allow the robot to slide parallel to $f_1$. In the worst case, when both contacts $O_1$ and $O_2$ lie along the edge of length $D$ very close to each other, the sliding distance allowed by $O_1$, $O_2$ is almost $D$. The third contacts $O_3$ can be classified into \emph{type-$\pi_1$ contacts} and \emph{type-$\pi_2$ contacts}, according to whether $g_3$ lies in $\pi_1(f_1)$ or in $\pi_2(f_1)$. As the robot slides preserving $O_1$ and $O_2$, the third contacts alternate between type-$\pi_1$ contacts and type-$\pi_2$ contacts. Say $f_1$ is horizontal and $\pi_1(f_1)$ lies to the left of $f_1$. Consider sliding the robot from left to right. From a type-$\pi_1$ contact position, the robot must slide right some positive distance to reach the next type-$\pi_2$ contact position. And from a type-$\pi_2$ contact position, the robot must slide right distance at least $d$ (with equality only if the robot is a parallelogram and $d=D$) to reach the next type-$\pi_1$ contact position. Thus, there are at most $2\lceil D/d\rceil$ possible third contact positions. Since $2\lceil D/d\rceil$ is a constant, we again get an overall bound of $O(n^2)$ for this case.

\begin{figure}
	\centering\includegraphics[width=14cm]{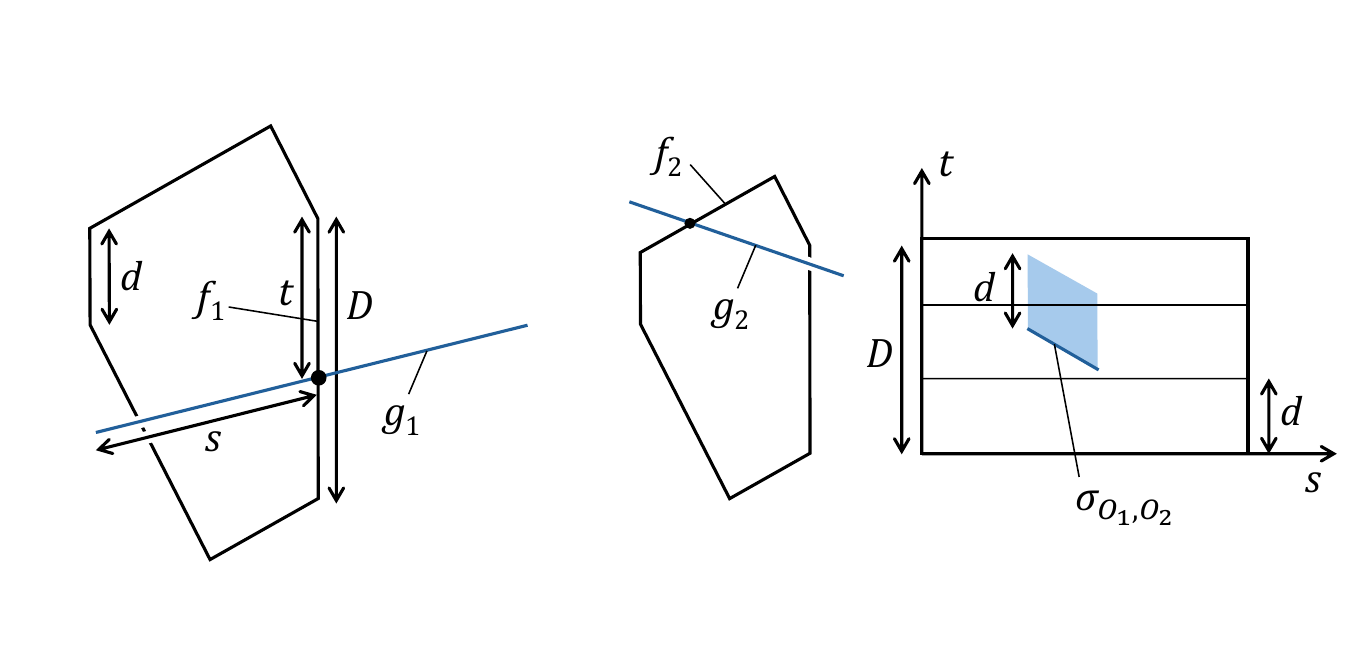}
	\caption{\label{fig_fully_parallel_edgePP}When the robot is a fully-parallel polygon, in the parametric plane of an edge contact $O_1$, the parallelogram-shaped region of height $d$ above or below the segment $\sigma_{O_1,O_2}$ is certainly illegal. Hence, we can divide the rectangle into horizontal strips and proceed as before.}
\end{figure}

Next, consider the cases in which $f_1$ is an edge of $B$, and each one of $f_2$, $f_3$ is either a vertex of $B$ or an edge of $B$ different from $f_1$ and from its opposite edge $f'_1$. Again, let $d$ and $D$ be the smaller and larger lengths of $f_1$, $f'_1$, respectively. Define the $st$-parametric plane $P_{O_1}$ by letting the $s$-coordinate be the distance between the left endpoint of $g_1$ and the point of contact $p$, and letting the $t$-coordinate be the distance between $p$ and the upper endpoint of $f_1$. The robot placements that realize $O_1$ correspond to a rectangle $\Gamma\subset P_{O_1}$. As before, for each choice of $f_2$, let $S_{O_1,f_2}$ be the set of all segments $\sigma_{O_1,(f_2,g_2)}\subset P_{O_1}$ over all valid choices of $g_2$. For each segment $\sigma_{O_1,O_2}\in S_{O_1,f_2}$, there is a parallelogram-shaped region of height $d$, either vertically above or vertically below $\sigma_{O_1,O_2}$ (depending on whether $f_2$ belongs to $\pi_1(f_1)$ or to $\pi_2(f_1)$), that is illegal, since each point in this region corresponds to a placement of $B$ that collides with $g_2$. See Figure \ref{fig_fully_parallel_edgePP}. Partition the rectangle $\Gamma$ into $\lceil D/d \rceil$ horizontal strips of width at most $d$. This operation might cut some segments in $S_{O_1,f_2}$ into (constantly many) subsegments. Within each strip there cannot be more than one point with a given $s$-coordinate that represents a legal double contact involving $O_1$ and $f_2$. Hence, we are interested only in either the lower or the upper envelope of the subsegments of $S_{O_1,f_2}$ (depending on $f_2$) that lie in that strip.

If one of $f_2$, $f_3$ is a vertex of $B$, then the segments in $S_{O_1,f_2}$ or $S_{O_1,f_3}$ (respectively) are pairwise disjoint, so Lemma \ref{lem_env} (Case 3) applies. 

Finally, suppose $f_1$, $f_2$, $f_3$ are three mutually nonparallel edges of $B$. Then at least one of the edges $f_i$ satisfies the property that the other two edges $f_j$, $f_k$ belong to opposite boundary parts $\pi_1(f_i)$, $\pi_2(f_i)$. In the corresponding parametric plane $P_{O_i}$, the segments $\sigma_{O_i,O_j}$ will have illegal regions above the segments, while the segments $\sigma_{O_i,O_k}$ will have illegal regions below the segments (or vice versa). Hence, Lemma \ref{lem_env} (Case 2) applies.

Finally, Corollary \ref{cor_poly_vvv} takes care of the case of three vertex contacts. Theorem \ref{thm_fully-parallel} is proven.
\end{proof}

\section{On the general position assumption}\label{sec_genpos}

We now present the rationale behind the general position assumption we have used throughout the paper.

Suppose the robot-and-obstacle configuration is not in general position. A perturbation might cause features of $\partial\mathcal F$ to disappear. Specifically, this might happen in places where expanded obstacles intersect at their boundaries but not their interiors. Hence we must be careful in the way we perturb the configuration. We mainly follow \cite{AS,AST}.

First, homothetically expand the obstacle configuration as a whole, by a factor $1+\varepsilon$ for $\varepsilon$ small enough. This operation is equivalent to homothetically shrinking the robot by a small amount, or equivalently, slightly deflating the expanded obstacles (cf.~remark after Theorem 4.2 in \cite{AS}).

This operation cannot cause any face or edge of $\mathcal F$ to disappear. The only features of $\mathcal F$ that might disappear are vertices of $\mathcal F$ that arose from contacts between two (or more) edges of expanded obstacles---the slight deflation of the expanded obstacles will cause such edges to move apart from each other. As a consequence, some edges and faces of expanded obstacles that were divided into sub-edges and sub-faces might merge together.

However, the number of vertices that disappear in this way is bounded by $O(n^2)$. Now consider an edge $e$ of $\mathcal F$ that undergoes a merge. The edge $e$ is part of an edge $e'$ of an expanded obstacle, which used to make contact with another expanded-obstacle edge $e''$. Charge the edge $e$ to the ordered pair $(e',e'')$. (We charge $e$ to the pair $(e',e'')$ and not to the disappeared vertex $v$, since $v$ might have had large degree.) Then each pair $(e',e'')$ receives at most two charges. Hence, the number of edges of $\mathcal F$ lost to merges it at most $O(n^2)$ as well.

Finally, consider a face $f$ of $\mathcal F$ that underwent a merge. Charge the face $f$ to the edge of $\mathcal F$ that moved away and caused the merge. Then each edge of $\mathcal F$ receives at most two charges. Hence, the number of faces of $\mathcal F$ lost to merges is bounded by $O(n^2)$ as well.

After this homothetic shrinking of the robot, we can safely apply an arbitrary small-enough perturbation that achieves general position. It will not cause any feature of the free space to disappear.

\section{Discussion and open problems}

The main open problem is to close the gap between $O(n^2\log n)$ and $\Omega(n^2\alpha(n))$ for the general case of a convex polyhedral robot.

Another open problem is to obtain more refined bounds that depend on the complexity of the robot, as well as on the number of obstacles, and not just on the total complexity of the obstacles, under the assumption that the obstacles are convex.

Suppose there are $k$ convex obstacles total complexity $n$. For the general case where the robot is an $r$-vertex polyhedron, Aronov and Sharir \cite{AS} proved an upper bound of $O(rnk\log k)$. It is easy to achieve a lower bound of $\Omega(nk)$ for any robot. For a flat triangular robot, the construction of Figure \ref{fig_triangle_lower_bound} can be modified to achieve a lower bound of $\Omega(nk\alpha(k))$ (see the details in \cite{AS}).  Halperin (personal communication) derived an upper bound of $O(nk)$ for the case of a segment-shaped robot. These are the currently best known bounds, as far as we know.

\paragraph{Acknowledgements.} Thanks to Danny Halperin for suggesting that I look into this problem and for useful conversations. Thanks to the anonymous referees for carefully reading previous versions of this paper and providing many helpful comments.

\paragraph{Data availability statement.} This paper has no associated data.

\bibliographystyle{plainurl}
\bibliography{new_robot_revised_arXiv}

\end{document}